\documentclass[conference]{IEEEtran}

\usepackage{multirow}
\usepackage{bm}
\usepackage{amsfonts}
\usepackage{amstext}
\usepackage{amsmath}
\usepackage{array}
\usepackage{graphicx}
\usepackage{CJK}
\usepackage{subfigure}
\usepackage{diagbox}

\ifCLASSINFOpdf
\usepackage{graphicx}
\usepackage{caption}

\else

\fi
\usepackage{amsthm,mathrsfs,amsfonts,dsfont}
\usepackage{epstopdf}
\usepackage{cite,url}
\usepackage{amssymb}
\usepackage{color,comment}
\usepackage{latexsym}
\usepackage{algorithm}
\usepackage{algorithmic}
\usepackage{multirow}
\usepackage{amsxtra}
\usepackage{xcolor}

\newtheorem{definition}{\bf Definition}

\newtheorem{proposition}{\bf Proposition}

\usepackage{array}
\usepackage{mdwmath}
\usepackage{mdwtab}
\usepackage{eqparbox}
\usepackage{subfigure}

\usepackage{algorithm}
\usepackage{algorithmic}
\usepackage{caption}
\makeatletter
\newcommand\fs@ruled@notop{\def\@fs@cfont{\bfseries}\let\@fs@capt\floatc@ruled
  \def\@fs@pre{}%
  \def\@fs@post{\kern2pt\hrule\relax}%
  \def\@fs@mid{\hrule height.8pt depth0pt \kern2pt}%
  \let\@fs@iftopcapt\iftrue}
\renewcommand\fst@algorithm{\fs@ruled@notop}
\makeatother

\DeclareCaptionLabelSeparator{period-quad}{.\quad}
\usepackage{amsmath}

\newlength{\aligntop}
\newlength{\alignbot}
\makeatletter

\renewenvironment{align}{%
\vspace{\aligntop}
\start@align\@ne\st@rredfalse\m@ne
}{%
\math@cr \black@\totwidth@
\egroup
\ifingather@
\restorealignstate@
\egroup
\nonumber
\ifnum0=`{\fi\iffalse}\fi
\else
$$%
\fi
\ignorespacesafterend%
\vspace{\alignbot}\par\noindent
}

\begin{document}

\addtolength{\textfloatsep}{-8pt}

\title{Network Formation in the Sky: Unmanned Aerial Vehicles for Multi-hop Wireless Backhauling\vspace{-0.3cm}}

\author{
\IEEEauthorblockN{Ursula Challita\IEEEauthorrefmark{1} and Walid Saad\IEEEauthorrefmark{2}}
\IEEEauthorblockA{\IEEEauthorrefmark{1}\small
School of Informatics,
The University of Edinburgh, Edinburgh, UK. Email: ursula.challita@ed.ac.uk.}
\IEEEauthorblockA{\IEEEauthorrefmark{2}\small Wireless@VT, Bradley Department of Electrical and Computer Engineering, Virginia Tech, Blacksburg, VA, USA. Email: walids@vt.edu.}\vspace{-0.7cm}
}

\maketitle
\begin{abstract}
\boldmath
To reap the benefits of dense small base station (SBS) deployment, innovative backhaul solutions are needed in order to manage scenarios in which high-speed ground backhaul links are either unavailable or limited in capacity. In this paper, a novel backhaul scheme that utilizes unmanned aerial vehicles (UAVs) as an on-demand flying network linking ground SBSs and the core network is proposed. The design of the aerial backhaul scheme is formulated as a network formation game among UAVs that seek to form a multi-hop backhaul network in the air. To solve this game, a myopic network formation algorithm which reaches a pairwise stable network upon convergence, is introduced. The proposed network formation algorithm enables the UAVs to form the necessary multi-hop backhaul network
in a decentralized manner thus adapting the backhaul architecture to the dynamics of the network. 
Simulation results show that the proposed network formation algorithm achieves substantial performance gains in terms of both rate and delay reaching, respectively, up to $40$\% and $41$\% compared to the formation of direct communication links with the gateway node (for a network with $15$ UAVs).
\end{abstract}
\IEEEpeerreviewmaketitle
\vspace{-0.13cm}
\section{Introduction}
\vspace{-0.13cm}
The dense and viral deployment of small base stations (SBSs) is expected to lie at the heart of emerging 5G networks~\cite{backhaul}. However, to reap the benefits of SBS deployment, innovative backhaul solutions are needed, as SBSs may be deployed in adverse locations and rural areas in which backhaul access is either inexistent or strictly limited in capacity~\cite{mona}.

Several approaches have been recently proposed for SBS backhauling~\cite{mona, backhaul, omid, backhaul_UAV_2}. Such solutions include wired and wireless backhauling to and from core network aggregators, cooperation through anchor base stations, and multi-hop over short-range links~\cite{mona, backhaul}. Nevertheless, existing solutions do not account for scenarios in which the high-speed ground backhaul is either congested, unavailable or limited in capacity. In such scenarios, the backhaul connectivity of SBSs can become a bottleneck thus degrading the performance of the radio access network.
Therefore, a novel paradigm shift of backhaul network design for 5G networks and beyond is needed. One promising solution for such scenarios is to deploy unmanned aerial vehicles (UAVs) for providing backhaul connectivity to the SBSs~\cite{mohammad_D2D} and~\cite{mingzhe}. Due to their rapid and flexible deployment capabilities, mobility, ability to fly above obstacles, and relatively low cost, UAVs have received considerable interest for different applications in wireless communications, and in particular, as communication relays~\cite{IMP_relays_potential, moving_relay_1, moving_relay_2}.


In this regard, the authors in~\cite{backhaul_UAV} propose a vertical fronthaul/backhaul framework based on UAVs and free-space optics communication. In~\cite{backhaul_UAV_2}, the authors consider the use of UAVs as relays for backhaul connectivity of high altitude balloons in case of temporary failed links.
The authors in~\cite{IMP_relays_potential} consider the formation of a multi-hop relay system based on UAVs in order to extend the communication range of the ground network. In~\cite{moving_relay_1} and~\cite{moving_relay_2}, the authors consider a mobile relay network model based on UAVs, where a UAV serves as a resilient moving relay among the SBSs. Although the use of UAVs as communication relays has been explored in the literature~\cite{IMP_relays_potential, moving_relay_1, moving_relay_2}, these works are restricted to ad hoc, rather than cellular networks. Moreover, one challenging area which remains relatively unexplored is the formation of the aerial graph that connects UAVs to the core network. Indeed, the existing prior art does not provide an efficient scheme, in terms of achievable rate and delay, for the formation of a multi-hop aerial network for SBS backhauling.

The main contribution of this paper is thus to introduce a novel backhaul scheme that utilizes UAVs as an on-demand flying network linking the SBSs and the core network in scenarios where ground backhaul
is either unavailable or limited in capacity. The design of the aerial backhaul network is formulated as a network formation game in which the agents are the UAVs. The objective of the proposed game is to allow the UAVs to autonomously learn which air-to-air (A2A) and air-to-ground (A2G) links to form in order to guarantee the connectivity of the SBSs to the core network. In particular, we consider that the UAVs form a multi-hop aerial network in which each UAV can individually select the path that connects it to the backhaul gateway node through other UAVs while optimizing its own utility. To solve this game, we propose a dynamic network formation algorithm that is guaranteed to reach a pairwise stable network upon convergence. Moreover, to ensure an efficient backhauling process between the UAVs, we incorporate the notion of virtual force fields into our dynamic algorithm. In essence, virtual forces allow the UAVs to adjust their location dynamically based on the links they want to form. We show that, using the proposed algorithm, the UAVs are able to self-organize into a stable tree structure rooted at the gateway node. To our best knowledge, \emph{this is the first work that exploits the framework of network formation games for the design of a UAV-based multi-hop backhaul network.} Simulation results show that the proposed approach achieves significant rate and delay improvements.

The rest of this paper is organized as follows. In Section II, we present the system model. Section III describes the proposed network formation game. The proposed network formation algorithm is given in Section IV. In Section V, simulation results are analyzed. Finally, conclusions are drawn in Section VI.
\vspace{-0.12cm}
\section{System Model}
\vspace{-0.1cm}
Consider a network composed of a set $\mathcal{S}$ of $S$ SBSs and a set $\mathcal{J}$ of $J$ UAVs. We consider a system in rural areas, hotspots, or ultra dense cellular areas in which SBSs are located at adverse locations (e.g., at lamp posts or street levels), and a ground backhaul network that connects the SBSs to the core network is either unavailable or limited in capacity. To overcome such bottleneck, we propose the use of UAVs as a temporarily aerial backhaul network for the SBSs. Specifically, UAVs serve as a bridge among the SBSs and relay the traffic to a nearby gateway node (with core network access) or as an intermediate relay point between different backhaul transceivers. 

In our model, the UAVs are initially located based on the deployment approach given in~\cite{mohammad} and each UAV $j$ serves a set of $\mathcal{S}_j$ SBSs. Packet forwarding is supported for both uplink (UL) and downlink (DL) directions via a frequency division duplex (FDD) model thus allowing the flow of traffic to/from the SBSs from/to the core network through a gateway node. We consider the availability of one gateway node $n$ in a given area and assume that at least one UAV has access to this gateway node. Given that the communication range of low-altitude platform (LAP) UAVs is typically limited to a few hundred meters, after which the signal quality deteriorates~\cite{UAV_survey}, the formation of a multi-hop aerial network becomes necessary to extend the communication range of the network and provide service to SBSs that are located at distant or hard to reach areas where infrastructure does not exist. Consequently, a communication link with the infrastructure is formed through either UAV-to-UAV multi-hop links or a UAV-to-infrastructure data link.



\subsection{A2G and A2A channel models}
For our proposed model, we consider that UAVs transmit over the sub-6 GHz band for the A2G and A2A links. We adopt the free-space path loss model, $\xi$, given by~\cite{hourani}:
\begin{align}
\xi (\mathrm{dB})= 20 \mathrm{log}_{10} (d_{o,d}) + 20 \mathrm{log}_{10} (f_c) - 147.55,
\end{align}
\noindent where $f_c$ is the system center frequency (in Hz) and $d_{o,d}=\frac{\Delta h_{o,d}}{\mathrm{sin}\theta_{o,d}}$, is the Euclidean distance between an origin node $o$ and a destination node $d$ (in $m$); $\Delta h_{o,d} = z_o-z_d$ is the altitude difference between $o$ and $d$ and $\theta_{o,d}$ is the elevation angle.
The use of a free space propagation model is validated by the fact that LAP UAVs fly at an altitude of $\sim$100m.

We consider a probabilistic LoS and non-line-of-sight (NLoS) links for the A2G propagation channel as done in~\cite{mohammad}. In such a model, NLoS links experience higher attenuations due to the shadowing and diffraction loss. Therefore, the adopted path loss model between UAV $j$ and SBS $s$, $L_{j,s}$, is given by:
\begin{align}\label{pathloss}
\hspace{-1.2cm}L_{j,s} =
\begin{cases}
\xi_{j,s} + \eta_{\mathrm{LoS}},  & \text{LoS link},\\
\xi_{j,s} + \eta_{\mathrm{NLoS}}, & \text{NLoS link}.
\end{cases}
\end{align}
\noindent where $\eta_{\mathrm{LoS}}$ and $\eta_{\mathrm{NLoS}}$ correspond to the additional attenuation factor added to the free space propoagation model and due to the NLoS connection links, respectively. Here, the probability of LoS connection depends on the environment, density and height of buildings, the location of the UAV and the SBS and the corresponding elevation angle between them. The LoS probability is given by~\cite{hourani}:
\begin{align}
P_{j,s}^{\mathrm{LoS}}=\frac{1}{1+C \mathrm{exp}(-D[\theta_{j,s}-C])},
\end{align}
\noindent where $C$ and $D$ are constants which depend on the environment (rural, urban, dense urban, or others) and $\theta_{j,s}=\mathrm{sin}^{-1}(\frac{\Delta h_{j,s}}{d_{j,s}})$ is the elevation angle. Clearly, the probability of NLoS is $P_{j,s}^{\mathrm{NLoS}}=1-P_{j,s}^{\mathrm{LoS}}$. Therefore, the average path loss between UAV $j$ and SBS $s$, $\overline{L}_{j,s}$, is given by:
\begin{align}\label{average_pathloss}
\overline{L}_{j,s}= P_{j,s}^{\mathrm{LoS}}\cdot L_{j,s}^{\mathrm{LoS}} + P_{j,s}^{\mathrm{NLoS}}\cdot L_{j,s}^{\mathrm{NLoS}},
\end{align}

For the A2A links, we consider LoS links between different UAVs that wish to form a link. Therefore, the path loss model between UAV $j$ and UAV $i$, $L_{j,i}$, will be $L_{j,i}=\xi_{j,i} + \eta_{\mathrm{LoS}}$.

Based on the given channel model, the average signal-to-interference-plus-noise ratio (SINR) of the A2G link between
an origin node $o$ and a destination node $d$ (which can represent the
link between UAV $j$ and SBS $s$ or UAV $j$ and the gateway node $n$) in the DL or UL direction, $\Gamma_{o,d}$, is given by:
\begin{align}\label{SNIR}
\Gamma_{o,d}=\frac{P_{o,d}\cdot h_{o,d}}{\sum_{q=1, q\neq o}^O I_{q,d}+\sigma^2},
\end{align}

\noindent where $P_{o,d}$ is the transmit power of the origin node $o$ (which can
represent UAV $j$ or SBS $s$) to the destination node $d$, $h_{o,d}=1/10^{\overline{L}_{o,d}/10}$ is the channel gain between $o$ and $d$, $\sigma^2$ is the Gaussian noise and $\sum_{q=1, q\neq o}^O I_{q,d}$ is the total interference power at the destination node $d$ from other neighboring origin nodes $q$ (UAVs in the DL or SBSs in the UL) that are transmitting on the same channel, where $I_{q,d}=P_{q,d}/10^{\overline{L}_{q,d}/10}$. Therefore, based on Shannon's capacity, the achievable data rate of the A2G link can
be defined as $R_{o,d}=B_{o} \mathrm{log}_2(1+\Gamma_{o,d})$, where $B_{o}$ is the transmission bandwidth of the origin node $o$.



For the A2A links, we consider orthogonal channel allocation among all UAVs and, hence, the signal-to-noise ratio (SNR) between UAVs $j$ and $i$ is given by $\Gamma_{j,i}=\frac{P_{j,i}}{10^{L_{j,i}/10}\cdot \sigma^2}$. The capacity of the A2A link is $R_{j,i}=B_{j} \mathrm{log}_2(1+\Gamma_{j,i})$, where $B_{j}$ is the transmission bandwidth of UAV $j$.

Therefore, the achievable end-to-end rate of UAV $j$ along a multi-hop path $p_j$ in the DL direction, $R_{j}^{\mathrm{DL}}(p_j)$, corresponds to the minimum of the rates achievable over $N$ hops, as given below~\cite{multihop_capacity}:
\begin{align}\label{multihop_rate}
R_{j}^{\mathrm{DL}}(p_j) = \min_{\substack{n=1, \cdots, N}} R_{j_k, j_{k+1}}^{\mathrm{DL}},
\end{align}

\noindent where $R_{j_k, j_{k+1}}^{\mathrm{DL}}$ corresponds to the rate over link $j_k j_{k+1}$ in the DL direction. Similarly for $R_{j}^{\mathrm{UL}}(p_j)$, the achievable rate in the UL direction along path $p_j$.

\subsection{Problem formulation}
Given this network, our objective is to form an aerial backhaul network that allows each UAV $j$ to be connected to the gateway node $n$ via at most \emph{one} path, denoted as $p_j$, whenever this path exists. To realize this, we consider the formation of a \emph{bidirectional tree structure} rooted at the gateway node $n$. We let $\beta_{j,i} =1$ if link $ji$ is formed between UAV $j$ and UAV $i$, and $0$, otherwise and $\alpha_{j,n} =1$ if link $jn$ is formed between UAV $j$ and the gateway node $n$, and $0$, otherwise. Therefore, the centralized optimization problem can be formulated as follows:
\begin{align}\label{obj}
\max_{\mathbf{\boldsymbol{\alpha}_{j,n}, \boldsymbol{\beta}_{j,i}}} \sum_{j=1}^J \phi_j(p_j(\alpha_{j,n}, \beta_{j,i})), 
\end{align}
\vspace{-0.5cm}
\begin{align}\label{cons_1}
\hspace{-1cm}\mathrm{s. t.} \;\;\;\; \sum_{i=1, i\neq j}^J \beta_{j,i} + \alpha_{j,n} \geq 1 \;\;\forall j, \;\;\;\;\;\;\;\sum_{j=1}^J \alpha_{j,n} \geq 1,
\end{align}
\vspace{-0.5cm}
\begin{align}\label{cons_3}
\sum_{j=1}^J \Big(\alpha_{j,n} + \beta_{j,i} \Big)=J,
\end{align}
\vspace{-0.4cm}
\begin{align}\label{cons_4}
\alpha_{j,n}\in \{0,1\}, \;\; \beta_{j,i}\in \{0,1\} \;\;\;\forall j,n.
\end{align}
\noindent where $\phi_j(p_j(\alpha_{j,n}, \beta_{j,i}))$ corresponds to the utility function of UAV $j$ along its path $p_j$. (\ref{cons_1}) guarantees the formation of at least one path for each UAV $j$ to the gateway node $n$ (via direct or multi-hop). The left-most constraint in (\ref{cons_1}) ensures that UAV $j$ is connected to at least another UAV $i$ in the network or to the gateway node $n$. The right-most constraint in (\ref{cons_1}) guarantees that at least one UAV $j$ is connected to the gateway node $n$. (\ref{cons_3}) limits the maximum number of formed edges in the network to $J$, the number of available UAVs. Thus, (\ref{cons_1}) and (\ref{cons_3}), avoid the formation of cycles/loops in the network and therefore guarantee the
formation of a tree structure rooted at the gateway node. (\ref{cons_4}) represents the feasibility constraints. 

Note that, although a fully centralized
approach can be used to form the aerial backhaul network, the need for a distributed solution is desirable for our problem as it has several advantages. For instance,
a centralized control system suffers from the single-point failure problem and hence can be a
bottleneck for communication and security.
On the other hand, a distributed approach does not rely on a
single controller which, if compromised (due to malicious attacks or failures), can disrupt the operation of the entire network.
Further, a centralized approach requires the controller to communicate with all UAVs at all time. This might not be feasible in case the UAVs belong to different operators. Moreover, it can yield significant overhead and complexity, namely in networks with a rapidly changing environment due to the mobility of UAVs or incoming traffic load.
Note also that, due to the change in the network topology, a UAV might not be always reachable from the controller. Given these reasons, a distributed approach for network formation is needed, as proposed next.


\vspace{-0.3cm}
\section{Network Formation Game}
\vspace{-0.1cm}
Our main objective is to provide a distributed approach that can model the interactions among UAVs that seek to form the aerial multi-hop backhaul network.
For this purpose, we adopt the analytical framework of network formation games~\cite{jackson_wolinsky, walid_networkformation} which
involves a number of independent decisions makers that interact with each others in order to form
a suited network graph that connects them. For our proposed game, the agents correspond to the set of UAVs and the action space of each UAV is defined as the set of links which UAV $j$ can delete or form. Therefore, we consider an \emph{undirected} graph $G(\mathcal{V}, \mathcal{E})$
with $\mathcal{V}$ being the set of all vertices ($J$ UAVs and gateway node $n$) that will be present in the graph and $\mathcal{E}$ the set of all edges (links) that connect different pairs
of nodes. Each undirected link $ji \in \mathcal{E}$ between two nodes $j$ and $i$ corresponds to the DL/UL traffic flow between these nodes. Given any network $G(\mathcal{V}, \mathcal{E})$, the path $p_j$ from UAV $j$ to the gateway node $n$ is defined as a sequence of nodes $j_1, \cdots, j_K$ (in $\mathcal{V}$) such that $j_1 = j$, $j_K = n$ and each
undirected link $j_k j_{k+1} \in \mathcal{E}$ for each $k \in \{1, \cdots ,K-1\}$.

Therefore, each UAV $j$ aims at optimizing its own utility by selecting an appropriate path that connects it to the backhaul gateway node through other UAVs. Subsequently, the UAVs can act as source nodes transmitting the received SBSs/gateway node packets to the gateway node/SBSs through one or more hops in the formed graph. The resulting network graph $G$ is highly dependent on the goals, objectives, and incentives of each UAV. For instance, the number of hops can have an impact on the end-to-end delay, scalability, and throughput and therefore, can affect the performance of the resulting network. Next, we define the proposed utility function for our game.

\subsection{Utility function}
The utility of each UAV $j$ is a function of the network topology and the set of links formed among different UAVs. To this end, we propose a utility function that captures key metrics such as rate, delay, and number of relayed packets.

\subsubsection{Achievable data rate}
To maximize the performance of the SBSs, each UAV $j$ aims at maximizing its end-to-end achievable data rate along its path $p_j$ in the DL and UL directions, denoted as $R_{j}^{DL}(p_j,G)$ and $R_{j}^{UL}(p_j,G)$ respectively.

\subsubsection{Number of relayed packets}
To provide incentives for UAVs to route each others packets, each UAV $j$ is given a positive utility equivalent to the number of packets it transmits/relays successfully to/from the
gateway node via DL and UL, $P_{j}^{DL}(p_j,G)$ and $P_{j}^{UL}(p_j,G)$, respectively. These packets correspond to the packets originating from the set $\mathcal{S}_j$ of SBSs connected to UAV $j$ and from all other UAVs that are connected to UAV $j$.

\subsubsection{Delay cost}
We define $\tau_{j}(p_j,G)$ as the average delay
over path $p_j=\{j_1, \dots, j_K\}$ from SBS $s$ connected via UAV $j$ to the core network (or vice versa) given by~\cite{walid_networkformation}:
\begin{align}\label{delay_eqn}
\tau_{j}(p_j\textrm{,}G)\textrm{=}\sum_{\textrm{(}j_k\textrm{,} j_{k+1}\textrm{)}\in q_j}\textrm{\Big(}\frac{\Psi_{j_k\textrm{,}j_{k+1}}}{2\mu_{j_k\textrm{,}j_{k+1}}\textrm{(}\mu_{j_k\textrm{,}j_{k+1}}\textrm{-}\Psi_{j_k\textrm{,}j_{k+1}}\textrm{)}}\textrm{+}\frac{1}{\mu_{j_k\textrm{,}j_{k+1}}}\textrm{\Big)}\textrm{,}
\end{align}
\noindent where $\Psi_{j_k,j_{k+1}}=\Lambda_{j_k} + \Delta_{j_k}$ is the total packet arrival rate (packets/sec)
traversing link $(j_k, j_{k+1}) \in p_j$ between UAV $j_k$ and UAV $j_{k+1}$
and originating from the set $\mathcal{S}_j$ of SBSs connected to UAV $j_k$, $\Lambda_{j_k}$ and from all other UAVs that
are connected to UAV $j_k$, $\Delta_{j_k}$ (considering the Kleinrock approximation~\cite{kleinrock}).
$\Lambda_{j}$ is defined as $\Lambda_{j}=\sum_{s\in \mathcal{S}_j}\lambda_s$,
where $\lambda_s$ corresponds to the average arrival rate of the traffic of SBS $s$ and $\Delta_{j_k}$ is defined as $\Delta_{j_k}=\sum_{i\in A_{j_k}}\Lambda_{i}$
where $A_{j_k}$ is the set of UAVs that have a link formed with UAV $j$. $\mu_{j_k,j_{k+1}}=R_{j_k, j_{k+1}}/\upsilon$
is the service rate over link $(j_k,j_{k+1})$ where $R_{j_k, j_{k+1}}$ is the rate of the direct transmission between UAV $j$ and UAV $j+1$ and $\upsilon$ is the packet length. According to (\ref{delay_eqn}), the delay will be infinite when $\mu_{j_k,j_{k+1}}<\Psi_{j_k,j_{k+1}}$.

\subsubsection{Total utility}
Hence, the utility function $U_j(p_j, G)$ of UAV $j$ along path $p_j$ for $\mathcal{S}_j\neq \emptyset$, is defined as:
\begin{multline}\label{utility}
U_j(p_j, G)=\Big(R_{j}^{DL}(p_j,G) + R_{j}^{UL}(p_j,G)\Big)\\
+ \delta_j \Big(P_{j}^{DL}(p_j,G) + P_{j}^{UL}(p_j,G)\Big)\\- \gamma_j \Big(\tau_{j}^{DL}(p_j,G) + \tau_{j}^{UL}(p_j,G) \Big),
\end{multline}
\noindent where $\delta_j$ and $\gamma_j$ are multi-objective weights.

Note, here that there is no incentive for any UAV $j$ to be disconnected from the gateway node, otherwise, its delay cost and thus its utility function would be infinite. Therefore, for any network formation algorithm, the resulting tree graph of our proposed game is always connected.

\subsection{Pairwise stability}
Given the fact that, in network formation games, the consent of two nodes is required to form a single link, the stability of the outcome can be more accurately characterized by considering bilateral deviations.
To satisfy this requirement, we consider the notion of \emph{pairwise stability} that was introduced in~\cite{jackson_wolinsky}.

\begin{definition}
\emph{A network $G$ is} pairwise stable \emph{with respect to the proposed utility function $U_j(p_j,G)$ if:\begin{enumerate}
            \item for all $ji \in \mathcal{E}$, $U_j(p_j,G)\geq U_j(p_j-ji, G-ji)$ and $U_i(p_i,G)\geq U_i(p_i-ji,G-ji)$, and
            \item for all $ji \notin \mathcal{E}$, if $U_j(p_j+ji,G+ji) > U_j(G)$ then $U_i(p_i+ji,G+ji)< U_i(G)$,
          \end{enumerate}
where $G-ji$ refers to deleting link $ji$ from $G$ and $G+ji$ refers to adding link $ji$ to $G$.}
\end{definition}
\begin{definition}
\emph{When a network $G$ is not pairwise stable, it is said to be} defeated \emph{by $G'$ if either $G'=G+ij$  and 2) is violated for $ij$, or if $G'=G-ij$ and 1) is violated for $ij$}.
\end{definition}

Therefore, a given backhaul graph is pairwise stable if there is no incentive for any UAV $j$ to break a link that is formed with another UAV $i$ (unilateral deviation) and no pair of UAVs $j$ and $i$ have an incentive to establish a new link (bilateral deviation). Under pairwise stability, one can ensure that each UAV $j$ will not change its link formation strategy and therefore guarantee the promised performance for other UAVs in the network, and more specifically, to those connected to it or belong to its path $p_j$. Moreover, given that the graph resulting from our proposed network formation game is always a tree structure, $G$ is pairwise stable if and only if no pair of UAVs can profitably deviate by \emph{simultaneously} breaking one link and forming
another. In other words, given UAVs $j$ and $i$ and any link $jk \in \mathcal{E}$, let $G'=G-jk+ji$, $p_j'=p_j-jk+ji$ and $p_i'=p_i+ji$ then:
\begin{align}
U_{j}(p_j,G) < U_j(p_j',G') \Rightarrow U_{i}(p_i,G) > U_i(p_i',G').
\end{align}

Note, however, that pairwise stable networks may not always exist. In particular, this occurs when each network is defeated by some adjacent network, and that these \emph{improving paths}
form cycles with no undefeated networks existing.

\begin{definition}
\emph{An} improving path \emph{is a sequence of networks $\{G_1, G_2, \cdots, G_k\}$ where each network $G_k$ is defeated by the subsequent network $G_{k+1}$.}
\end{definition}

\begin{definition}
\emph{A} cycle \emph{is an improving path $\{G_1, G_2, \cdots, G_k\}$ such that $G_1=G_k$.}
\end{definition}

Consequently, a network is pairwise stable if and only if it has no improving paths emanating from it. In fact, for any network graph $G$, there exists either a pairwise stable network (or more) or a cycle of networks~\cite{pairwise_existence}. For network formation games, given that the strategy space is typically discrete, it is customary to characterize pairwise stable networks using an algorithmic approach, as the derivation of closed-form equilibrium policies is not possible~\cite{jackson_wolinsky}. As such, next, we propose a dynamic network formation algorithm that is guaranteed to reach a pairwise stable network upon convergence.

\section{Distributed Dynamic Network Formation}
\vspace{-0.1cm}
To allow UAVs to adapt their location based on the resulting formed graph, we first incorporate the notion of virtual (artificial) force field in our proposed network formation algorithm.

\subsection{Virtual force field}
Given the initial locations of the UAVs, the formation of an aerial backhaul network might not be feasible in case UAVs are located outside each others' communication range. Therefore, to adjust the location of UAVs based on the links they want to form, a dynamic and self-organizing approach that allows adaptation to the dynamics of the network, is necessary. In this regard, we adopt the notion of virtual forces for UAVs~\cite{virtual_force_robotics}. A virtual force field allows a UAV to adjust its location by exerting forces of attraction and repulsion towards other UAVs. For our model, we consider the SNR as a metric for updating the value of the virtual force vector. In particular, to guarantee an efficient backhauling process, a minimum threshold value of SNR, denoted as $\widehat{\Gamma}$, should be achieved over each of the formed links. This in turn allows the determination of the maximum distance between UAVs $j$ and $i$, $d^{\mathrm{max}}_{j,i}$.

\begin{proposition}
\emph{To guarantee a minimum threshold value of SNR between UAV $j$ and UAV $i$, the distance between the two UAVs should not exceed $d_{j,i}^{\mathrm{max}}$, which is defined as:}
\begin{align}
d_{j,i}^{\mathrm{max}}=\sqrt{\frac{P_{j,i}}{\widehat{\Gamma} \cdot \sigma^2 \cdot 10^{\eta_{\mathrm{LoS}}/10} \cdot (\frac{4 \pi f_c}{c})^2}},
\end{align}
\noindent \emph{where $P_{j,i}$ is the transmit power from UAV $j$ to UAV $i$ and $c$ is the speed of light.}
\end{proposition}
\begin{proof}
The derivation of the expression of $d_{j,i}^{\mathrm{max}}$ follows from the definition of the SNR between UAV $j$ and UAV $i$, $\Gamma_{j,i}$.
\end{proof}

In fact, a virtual force can be expressed by a polar coordinate notation $(r, \theta)$ where $r$ is its magnitude and $\theta$ its orientation angle. It can act as an attractive or a repulsive force, adapting to the actions of each UAV. For our proposed model,
we consider an attractive virtual force from UAV $j$ towards UAV $i$, when both UAVs agree on the formation of link $ji$ but are out of each other's communication range. Therefore, the attractive force vector from UAV $j$ towards UAV $i$, $\overrightarrow{\boldsymbol{F}}_{j,i}^A$, is expressed as:
\begin{align}
\overrightarrow{\boldsymbol{F}}_{j,i}^A= \Big(u_A \cdot (d_{j,i}-d^{\mathrm{max}}_{j,i}), \theta_{j,i}\Big),
\end{align}
\noindent where $u_A$ corresponds to the virtual force attractive coefficient and $d_{j,i}$ is the Euclidean distance
between UAV $j$ and UAV $i$. 
On the other hand, a repulsive force is exerted from UAV $j$ towards UAV $i$, in case link $ji$ is deleted and is expressed as:
\begin{align}
\overrightarrow{\boldsymbol{F}}_{j,i}^{R1}= \Big(u_{R1} \cdot (d_{j,i}-d^{\mathrm{max}}_{j,i}), \theta_{j,i}+\pi \Big),
\end{align}
\noindent where $u_{R1}$ is the virtual force repulsive coefficient and $d_{j,i}$ is the Euclidean distance based on initial locations. Moreover, for physical collision avoidance between different UAVs, we define the following repulsive force from UAV $j$ towards UAV $i$:
\begin{align}
\overrightarrow{\boldsymbol{F}}_{j,i}^{R2}= \Big(u_{R2} \cdot \frac{1} {d_{j,i}}, \theta_{j,i}+\pi \Big),
\end{align}
\noindent where $u_{R2}$ corresponds to the virtual force repulsive coefficient for collision avoidance. Therefore, the total virtual force exerted from UAV $j$ on UAV $i$ can be written as:
\begin{align}
\overrightarrow{\boldsymbol{F}}_j = \sum_{i=1, i\neq j}^J \overrightarrow{\boldsymbol{F}}_{j,i}^A + \sum_{i=1, i\neq j}^J \overrightarrow{\boldsymbol{F}}_{j,i}^{R1} + \sum_{i=1, i\neq j}^J \overrightarrow{\boldsymbol{F}}_{j,i}^{R2},
\end{align}

In our model, we consider that UAVs broadcast their initial locations at $t=0$ and, hence, they can compute the corresponding virtual force vector even if they are not within each other's communication range. Therefore, given the strategies of each UAV $j$, its corresponding location is updated as follows:
\begin{align}
x_j'=x_j+\overrightarrow{\boldsymbol{F}}_{j}^x, \;\;y_j'=y_j+\overrightarrow{\boldsymbol{F}}_{j}^y, \mathrm{and}\;\;z_j'=z_j+\overrightarrow{\boldsymbol{F}}_{j}^z,
\end{align}
\noindent where $x_j$ and $x_j'$ are the initial and updated x-coordinate of UAV $j$, and $\overrightarrow{\boldsymbol{F}}_{j}^x$ is the x-component of $\overrightarrow{\boldsymbol{F}}_{j}$.
Consequently, this location update procedure improves the achievable data rate for each UAV along its path and thus ensures an efficient backhauling process.

\subsection{Dynamic network formation algorithm}
Taking into account the location update of UAVs based on the defined virtual forces, we propose a myopic dynamic network formation algorithm. In particular, myopic agents update their strategic decisions considering only the current state of the network without taking into account the future evolution of the network. To ensure the formation of a tree network architecture, link addition can be seen as link replacement and thus the strategy space of UAV $j$ can be regarded as either a delete operation or a \emph{replace} operation using which UAV $j$ replaces its previously connected link with its parent node (if it exists) with a new link. Let $\mathcal{W}$ denote the set of possible nodes with which UAV $j$ can possibly form or delete a link. We refer to $w \in \mathcal{W}$ as the \emph{activated} node which corresponds to any of the other $(J-1)$ UAVs or the gateway node $n$. The adopted rules for the formation of the undirected network graph are:
\begin{enumerate}
  \item UAV $j$ can add a link with node $w$ if both nodes $j$ and $w$ agree to add this link i.e., link addition is bilateral. Link $jw$ is formed via a link replacement strategy if $U_j(p_j-jl+jw, G-jl+jw) > U_j(p_j, G)$ and $U_w(p_w+jw, G-jl+jw)> U_w(p_w,G)$ where node $l$ corresponds to the parent node of UAV $j$ (if it exists). 
  \item UAV $j$ can delete link $jw$ if $U_j(p_j-jw, G-jw) > U_j(p_j, G)$ i.e., link deletion can be unilateral.
  \item Link replacement or deletion do not occur simultaneously.
\end{enumerate}
Note that the gateway node is considered to be a passive agent in our game.

For our network formation dynamics, we consider initially a star topology for $G_0$. Each iteration of our proposed algorithm consists of $J$ rounds during which the UAVs engage in the network formation game in an arbitrary but sequential order. At a given round, UAV $j$ chooses randomly (following a uniform distribution) another node $w$ and takes an action with respect to $w$. Following the network formation rules,
if link $jw$ exists between the two nodes, then node $j$ can delete this link if its beneficial for it. If link $jw$ is deleted, a repulsive force $\overrightarrow{\boldsymbol{F}}_{w,j}^R$ is exerted from UAV $w$ towards UAV $j$ thus returning UAV $w$ to its initial location, in case of location update during previous iterations. On the other hand, if link $jw$ does not exist, then UAV $j$
can split from its parent node $l$ and add link $jw$, if such a change is beneficial for both UAV $j$ and the activated node $w$. Here, both nodes $j$ and $w$ can communicate with each other via a direct temporarily communication link that is established in order to decide whether link $jw$ should be formed. Note that an attractive force $\overrightarrow{\boldsymbol{F}}_{j,w}^A$ is exerted from UAV $j$ towards node $w$ in case the corresponding two nodes are not within each other's communication range.
If, at the end of the round, both nodes agree on the formation of link $jw$, UAV $j$ updates its location to the current position. Otherwise, a repulsive force $\overrightarrow{\boldsymbol{F}}_{j,w}^R$ is exerted from UAV $j$ towards node $w$, thus returning UAV $j$ to its initial position at the beginning of this round. 
Note that $\overrightarrow{\boldsymbol{F}}_{j,w}^A$ and $\overrightarrow{\boldsymbol{F}}_{j,w}^R$ are exerted only when node $w$ is not the gateway node. At the end of each round, UAV $j$ and the activated node $w$ update their corresponding location and path and broadcast such information to all other UAVs. After the convergence of the network formation algorithm, the UAVs are connected through a tree structure rooted at the gateway node. Consequently, data packets from/to the SBSs to/from the core network can now be transmitted using the resulting formed network tree structure $G_{\textrm{final}}$. The convergence complexity of our proposed myopic network formation algorithm is $O(J^2)$. A summary of the proposed algorithm is given in Algorithm~\ref{algorithm}.

Given the definition of pairwise stability and the proposed network formation rules, it can be clearly seen that, if the network formation process converges to a final network $G$, then $G$ must be pairwise stable. However, proving the convergence of the network formation rules is challenging. In fact, if a pairwise stable network does not exist, then the proposed algorithm would involve cycles of networks which are randomly visited over time~\cite{pairwise_existence}. Therefore, using simulation, we show in the following section that our proposed algorithm will converge.


%

\begin{algorithm}[t!] \scriptsize
\caption{Proposed network formation algorithm.}
\label{algorithm}
\begin{algorithmic}[t!]
\STATE \textbf{Initialization:}\\ Consider initially a star network $G_0$ where each UAV $J$ is connected to the gateway node via a direct link.\\
\vspace{0.2cm}
\STATE \textbf{Myopic network formation:}
\WHILE{$G$ has not yet converged to a stable network,}
\vspace{0.05cm}
\STATE \emph{In a random but sequential order, the UAVs engage in the network formation game.}\\
\vspace{0.05cm}
\STATE \textbf{Step 1.} UAV $j$ activates another node $w$, in a random fashion but following a uniform distribution.
\vspace{0.1cm}
\IF {$jw \in \mathcal{E}$}
\STATE \textbf{Step 2.} UAV $j$ deletes link $jw$ if $U_j(p_j-jw, G-jw) > U_j(p_j, G)$.
\IF {link $jw$ is deleted}
\STATE \textbf{Step 3.} A repulsive virtual force $\overrightarrow{\boldsymbol{F}}_{w,j}^R$ is exerted from UAV $w$ towards UAV $j$ thus returning UAV $w$ to its initial location.
\ENDIF
\ENDIF
\vspace{0.1cm}
\IF {$jw \notin \mathcal{E}$}
\IF {UAV $j$ and node $w$ are not within each other's communication range}
\STATE \textbf{Step 4.} An attractive virtual force $\overrightarrow{\boldsymbol{F}}_{j,w}^A$ is exerted from UAV $j$ towards node $w$ thus updating the location of UAV $j$.
\ENDIF
\STATE \textbf{Step 5.} UAV $j$ establishes a temporarily communication link with node $w$.
\IF {$U_j(p_j-jl+jw, G-jl+jw) > U_j(p_j, G)$ and $U_w(p_w+jw, G-jl+jw)> U_w(p_w, G)$ where node $l$ corresponds to the parent node of UAV $j$
  (if it exists)}
\STATE \textbf{Step 6.} Link $jw$ is formed via a link replacement strategy.
\ELSE
\STATE \textbf{Step 7.} A repulsive virtual force $\overrightarrow{\boldsymbol{F}}_{j,w}^R$ is exerted from UAV $j$ towards node $w$ thus returning UAV $j$ to its initial location, in case of location update at Step 6.
\ENDIF
\ENDIF
\STATE \textbf{Step 8.} UAV $j$ and node $w$ broadcast their updated locations and paths to all other nodes in the network.

\ENDWHILE
\end{algorithmic}
\end{algorithm}

\section{Simulation Results and Analysis}
\vspace{-0.1cm}
For our simulations, we consider a $5$~km $\times$ $5$~km square
area in which we randomly deploy a number of SBSs and UAVs. Table~\ref{table_parameters}
summarizes the main simulation parameters. Note that the bandwidth per UAV is defined as the ratio of the total channel bandwidth $B$ to the number of UAVs. All statistical results are averaged over $1000$ independent runs.

\begin{table}[t!]\footnotesize
\setlength{\belowcaptionskip}{0pt}
\setlength{\abovedisplayskip}{3pt}
\captionsetup{belowskip=0pt}
\newcommand{\tabincell}[2]{\begin{tabular}{@{}#1@{}}#1.6\end{tabular}}
 \setlength{\abovecaptionskip}{2pt}
 \renewcommand{\captionlabelfont}{\small}
\caption[table]{\scriptsize{\\SYSTEM PARAMETERS}}\label{table_parameters}
\centering
\tabcolsep=0.11cm 
\scalebox{0.99}{
\begin{tabular}{|c|c|c|c|}
\hline
\textbf{Parameters} & \textbf{Values} & \textbf{Parameters} & \textbf{Values} \\
\hline
Max transmit power $(P_o)$ & 20 dBm & $\eta_{\mathrm{LoS}}$ & 5 dB\\
\hline
SNR threshold $(\widehat{{\Gamma}})$ & -4 dB & $\eta_{\mathrm{NLoS}}$ & 20 dB\\
\hline
Speed of light ($c$) & $3\times10^8$ m/s & $C$ & 11.9\\ \hline
Channel bandwidth $(B)$ & 40 MHz & $D$ & 0.13 \\
\hline
Noise variance $(\sigma^2)$ & -90 dBm &$u_A$&1\\
\hline
Carrier frequency $(f_c)$& 2 GHz & $u_{R1}, u_{R2}$ & 10\\
\hline
Packet arrival rate $(\lambda_s)$& (0, 1)&Packet size $(\upsilon)$ & 2000 bits \\
\hline
\end{tabular}
}
\vspace{-0.24cm}
\end{table}

Fig.~\ref{snapshot} shows a snapshot of the tree graph resulting from the proposed algorithm for a network with $J=10$ randomly deployed UAVs. From Fig.~\ref{snapshot}, we can see that most of the UAVs that are located far from the gateway engage in a multi-hop transmission with other UAVs that are located closer to the gateway thus extending the communication range of the network. Moreover, from this snapshot, we can see that the UAVs select their paths not only based on distance but also on the number of hops and traffic over a given path. For instance, UAV $9$ connects to UAV $6$, although UAV $7$ is closer. This is due to the fact that the path for UAV $9$ along UAV $7$ involves $5$ hops and is more congested as compared to $3$ hops and less traffic when connected to UAV $6$. This in turn decreases the latency along its path and thus improves its utility. From Fig.~\ref{snapshot}, we can also see the effect of the virtual force vector on the location of the UAVs. For instance, UAVs $3$ and $5$ adjust their initial location in order to guarantee an efficient communication link with UAV $6$. Here, note that one could deploy more UAVs in case the location update of a particular UAV causes severe degradation in the A2G link connecting it to its serving SBSs.

\begin{figure}[t!]
  \begin{center}
  \vspace{-0.5cm}
    \includegraphics[scale=0.33]{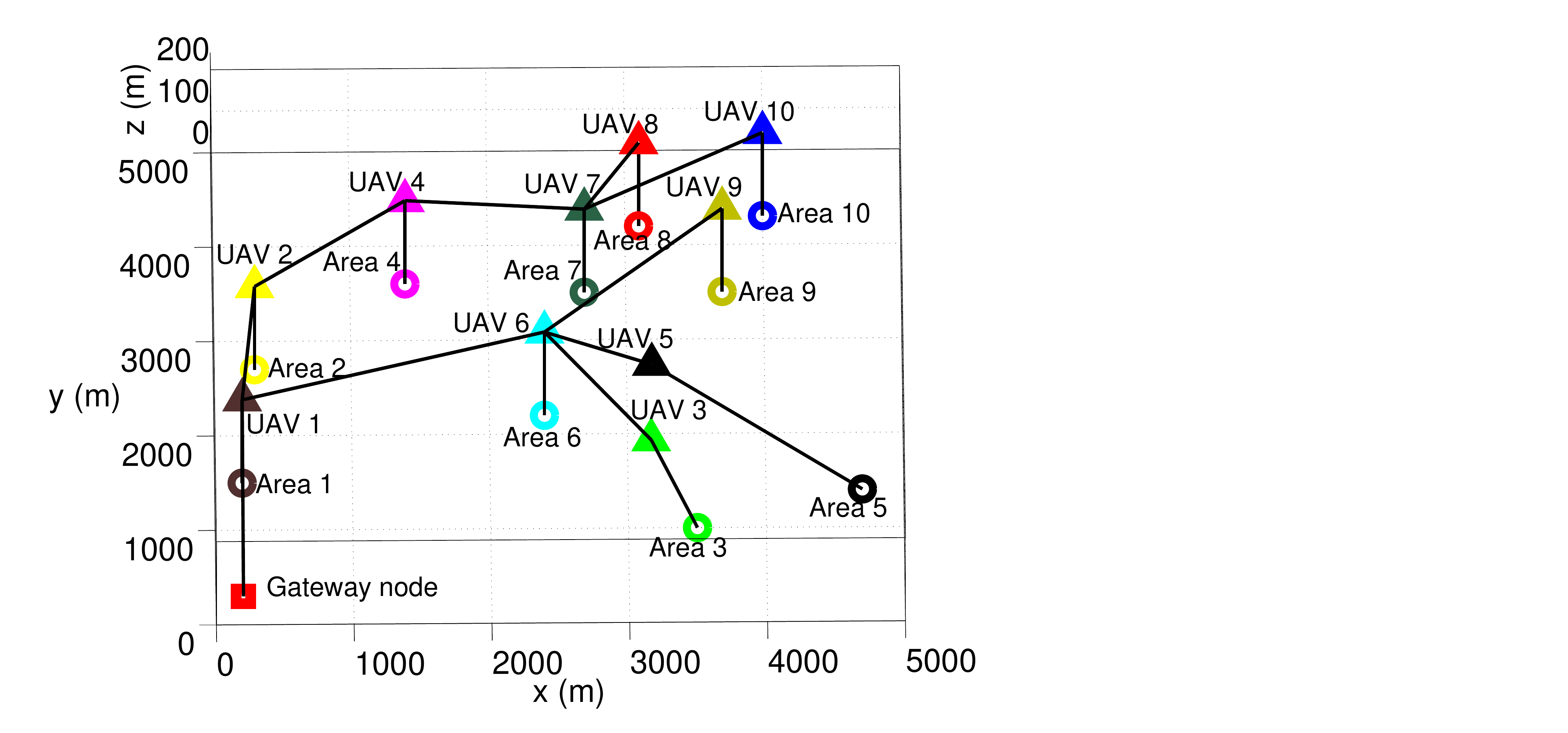}
   \caption{Snapshot of a tree graph formed using the proposed algorithm for a network with $J=10$ randomly deployed UAVs. Circles represent target areas having one or multiple SBSs.}\label{snapshot}
     \vspace{-0.6cm}
  \end{center}
\end{figure}

\begin{figure}[t!]
  \begin{center}
  \vspace{-0.51cm}
    \includegraphics[scale=0.23]{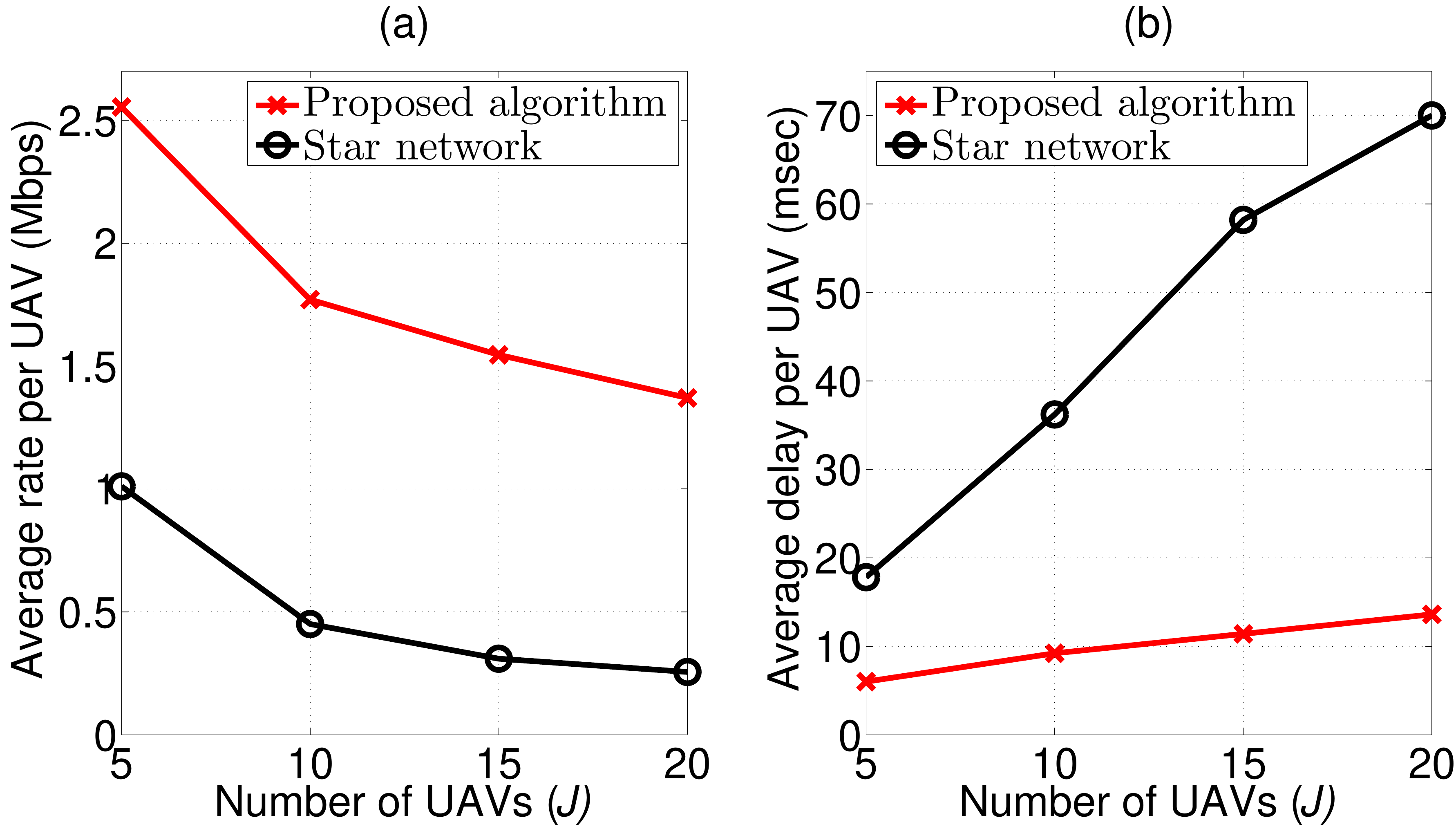}
   \caption{Performance assessment of the proposed network formation algorithm in terms of average (a) rate and (b) delay per UAV as compared to the star network, for different number of UAVs.}\label{performance}
     \vspace{-0.5cm}
  \end{center}
\end{figure}

Fig.~\ref{performance} shows the average achievable rate and delay per UAV of the resulting network for our proposed scheme and the direct transmission approach considering a star topology. From Fig.~\ref{performance}, we can see that, at all network sizes, the
proposed network formation algorithm yields significant
performance gains in terms of both rate and delay reaching, respectively, up to $40$\% and $41$\% relative to the star network (for a network with $15$ UAVs). The reason for this gain stems from the fact that multi-hop transmission allows UAVs having bad channel conditions with the gateway node to form links with other UAVs having better channel conditions. Here, note that the rate of the A2A links is higher than that of the A2G links due to the availability of a LoS communication links between different UAVs as well as the orthogonal channel allocation. Therefore, although more hops are formed, the average achieved rate over the multi-hop path is improved as compared to a direct link having weaker channel conditions. This in turn results in a higher service rate and thus a lower delay over the formed path. Here, note that the transmission bandwidth of each UAV is a function of the number of UAVs in the network. This in turn justifies the decrease in the average rate per UAV for both schemes as the number of UAVs increases. Fig.~\ref{performance} (b) demonstrates that, although the delay for both schemes increases as the number of UAVs in the network increases from $5$ to $20$, the speed at which the delay increases for our proposed scheme ($12.6$\%) is much smaller compared to that of the star network ($29.3$\%). This is due to the fact that, for a given UAV $j$, the number of possible paths to the gateway node increases as the number of UAVs increases.



Fig.~\ref{iteration} shows the minimum, average, and maximum number of iterations needed till convergence of our proposed network formation algorithm as the number of UAVs increases. From Fig.~\ref{iteration}, we can see that our proposed network formation algorithm converges after a number of iterations and therefore a stable graph is reached. Moreover, we can note that as the number of UAVs increases, the total number of iterations required for the convergence of the algorithm increases. This result is due to the fact that, as $J$ increases, the number of possible activated nodes $w$ for a particular UAV $j$ increases, and, thus, more actions (i.e., iterations) are required prior to convergence. For instance, the minimum, average and maximum number of iterations vary, respectively, from $4$, $7.2$, and $23$ at $J=5$ UAVs up to $18$, $81$, and $170$ at $J=20$ UAVs. Here, it is worth noting that practical UAV-based backhaul solutions will typically use only a relatively small number of UAVs, thus the convergence time resulting from our approach is practically reasonable.

\begin{figure}[t!]
  \begin{center}
  \vspace{-0.4cm}
    \includegraphics[scale=0.24]{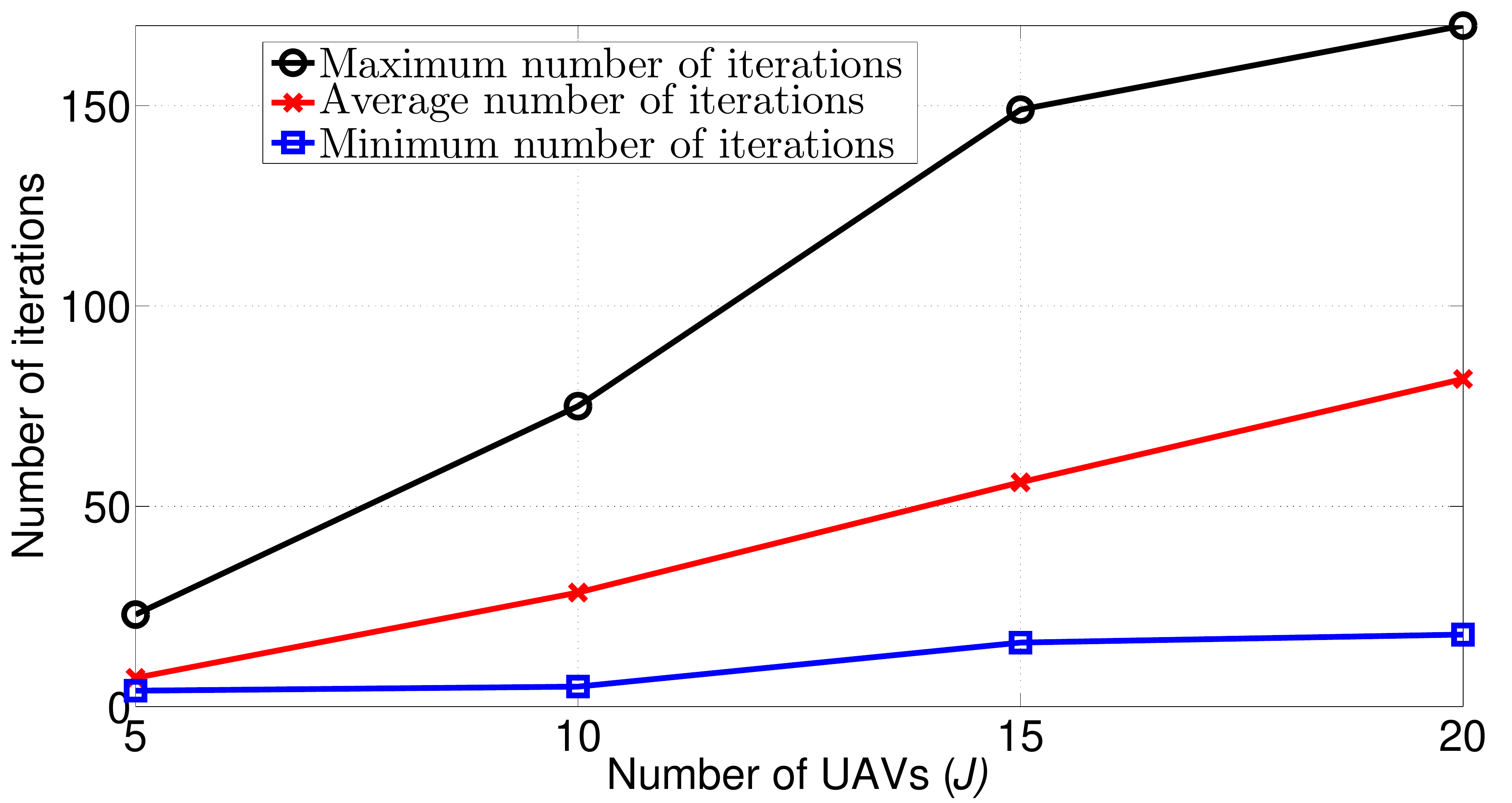}
   \caption{Minimum, average, and maximum number of iterations till convergence as a function of the number of UAVs $J$ in the network.}\label{iteration}
     \vspace{-0.5cm}
  \end{center}
\end{figure}


\section{Conclusion}
\vspace{-0.1cm}
In this paper, we have proposed a novel UAV-based backhaul network design for wireless networks. We have formulated the problem as a network formation game in which the UAVs seek to form a multi-hop aerial network that connects SBSs to the core network. In particular, each UAV can take an individual decision to optimize its utility by exploiting the possible paths that connects it to the gateway node. To solve the game, we have proposed a distributed myopic algorithm which is guaranteed to reach a pairwise stable network if converged. Simulation results have shown that the proposed approach yields significant performance gains in terms of delay and rate.

\def\baselinestretch{0.8}
\bibliographystyle{IEEEtran}
\bibliography{references}
\end{document}